\documentclass{article}

\usepackage{arxiv}
\usepackage[utf8]{inputenc} % allow utf-8 input
\usepackage[T1]{fontenc}    % use 8-bit T1 fonts
\usepackage{hyperref}       % hyperlinks
\usepackage{url}            % simple URL typesetting
\usepackage{booktabs}       % professional-quality tables
\usepackage{amsfonts}       % blackboard math symbols
\usepackage{nicefrac}       % compact symbols for 1/2, etc.
\usepackage{microtype}      % microtypography
\usepackage{lipsum}
\usepackage{graphicx}
\graphicspath{ {./images/} }
\usepackage{ulem}
\usepackage{diagbox}
\usepackage{comment}
\usepackage{nicematrix}
\usepackage{scrextend}
\usepackage{float}
 \usepackage{lscape}

\usepackage{amsmath,amssymb,latexsym}
\usepackage{amsthm}
\usepackage{stmaryrd}
\newtheorem{theorem}{Theorem}[section]

\newtheorem{lemma}[theorem]{Lemma}
\newtheorem{proposition}[theorem]{Proposition}

\theoremstyle{definition}
\newtheorem{definition}[theorem]{Definition}
\newtheorem{remark}[theorem]{Remark}

\newtheorem{example}[theorem]{Example}

\newtheorem*{notation}{Notation}

\usepackage[dvipsnames]{xcolor}
\usepackage{tikz}
\usepackage{enumitem}
\usepackage{lineno}
\newcommand{\FF}{\mathbb{F}}
\newcommand{\LL}{\mathbb{L}}
\newcommand{\KK}{\mathbb{K}}
\newcommand{\cv}{{\bf c}}
\newcommand{\Cc}{\mathcal{C}}
\newcommand{\Dc}{\mathcal{D}}
\newcommand{\Ul}{U_{\mathbb{L}}}
\newcommand{\defegal}{\overset{{\rm def}}{=}}
\newcommand{\Ww}{\text{w}}

\title{MacWilliams identities for the generalized rank weights}

\author{
 Julien Molina \\
  Université Grenoble Alpes\\
  Institut Fourier\\
  CS 40700, 38058 Grenoble cedex 9 \\
  \texttt{julien.molina@univ-grenoble-alpes.fr} \\
}

\begin{document}
\maketitle
\begin{abstract}
We study the generalized rank weight distribution of a linear code. First, we provide a MacWilliams-type identity which relates the distributions of a code and its dual. Then, we give a formula for the enumerator polynomial. Finally, we explicitly compute the distribution of an MRD code.
\end{abstract}

% keywords can be removed

\tableofcontents

\section{Introduction}
Linear codes over finite fields can be endowed with different metrics to study some of their properties. Oftentimes, metrics are constructed in the following way. Consider a weight wt on a code $\Cc$, which is nothing but a map wt $ : \Cc \longrightarrow \mathbb{N}$, from the code into the set of nonnegative integers. This weight induces a metric $d$ defined on the code by $d : \Cc \times \Cc \longrightarrow \mathbb{N}, (\cv_1,\cv_2) \longmapsto \text{wt}(\cv_1-\cv_2)$. Then, we can define the \textit{minimum weight} or the \textit{minimum distance} of the code $\Cc$ by $d_{min}(\Cc) \defegal \underset{\cv_1\neq \cv_2 \in \Cc}{\min} \{ d(\cv_1,\cv_2)\}$. Using the linearity of the code, we can rewrite it as $d_{min}(\Cc) = \displaystyle \min_{\substack{\cv\in \Cc \\ \cv\neq 0}} \{ \text{wt}(\cv)\}$. Once a code is equipped with a certain distance, one may want to compute how many codewords of a fixed weight there are in the code. In other words, we want to know the complete weight distribution of a code and also if this distribution might be linked to the weight distribution of the dual code. Those questions are relevant for practical means since the more we know the distribution, the better we get information on the code's error-correcting and error-detecting capacities.

Traditionally, codes were endowed with the Hamming metric. In this case, the weight considered is $\Ww \text{t}_H(\cv) \defegal \# \{i \mid c_i \neq 0\}$, where $\cv=(c_1,\dots,c_n) \in \Cc$. For this metric, the aforementioned questions are completely answered. Indeed, in \cite{MacWiliams}, MacWilliams demonstrated that the Hamming weight distribution of a code is linearly linked to the Hamming weight distribution of its dual ; this relationship was also rewritten in terms of their weight enumerator and is now known as \textit{MacWilliams identities}. Later, in \cite{Klove}, Kl{\o}ve gave a closed-form formula for the weight distribution of a code. Also, Tolhuizen demonstrated in \cite{Tolhuizen} that the Hamming weight distribution of a maximum distance separable (MDS) code does not depend on the code itself but only on its length, its dimension and the size of the alphabet.

Recently, linear codes have increasingly been considered with respect to the rank metric.
Rank metric for linear codes was first introduced and used independently by Gabidulin \cite{Gabidulin} and Roth \cite{ROTH}. This metric comes from the following weight. Let us pick a code $\Cc \subset \FF_{q^m}^n$ of length $n$ over $\FF_{q^m}$. Let us fix an $\FF_q$-basis $\mathcal{B}$ of $\FF_{q^m}$. For any vector $\cv = (c_1,\ldots,c_n)\in \FF_{q^m}^n$, let  $M_{\mathcal{B}}(\cv)\in {\rm M}_{m\times n}(\FF_q)$ be the matrix whose $i$th column is the vector of coordinates of $c_i$ with respect to $\mathcal{B}$. The \textit{rank weight} wt$_R$ of a vector $\cv\in \Cc$ is thus the rank of $M_{\mathcal{B}}(\cv)$. We refer to the recent survey \cite{FnT} and references therein for the state-of-the-art of results on codes with respect to the rank metric and their applications to network coding and cryptography. With respect to this metric, the same questions can arise as to counting the number of codewords with a fixed rank weight and whether MacWilliams identities hold. In this case, lots of questions are answered. In \cite[Theorem 1 and Proposition 2, Section III]{GadouleauYan}, Gadouleau and Yan proved a MacWilliams identity for the rank metric, in other words, there exists a linear relationship between the rank weight distributions of a code and the dual code. This relationship also exists for the weight enumerators using the $q$-transform. Years before, in \cite{Delsarte}, Delsarte demonstrated a version of MacWilliams identities and inequalities in the context of the normed space of all bilinear forms between vector spaces over $\FF_q$ endowed with the rank norm. His results were then applied to coding theory. Also in \cite[Proposition 3, Section III]{GadouleauYan}, Gadouleau and Yan showed that the rank weight distribution of a maximum rank distance (MRD) code is completely determined by its dimension, its length and the degree of the field extension. Moreover, in \cite{Ravagnani}, Ravagnani showed all the same MacWilliams identities and results on MRD codes but in the context of Delsarte rank metric codes, where codes are seen as subspaces of matrices with entries in $\FF_q$. In this same context, Shiromoto \cite{Shiromoto} proved the same results but using the techniques of $q$-(poly)matroids. Furthermore, rank metric codes studied over rings also satisfy MacWilliams identities, as demonstrated in \cite{BlancoBoixAlGaloisRings,
BlancoBoixAlChainRing}.

For some years, the sum-rank metric has been using and it can be seen as a generalization of both Hamming and rank metrics. This metric appears to have interesting applications for instance in the study of space-time codes or in the context of secure multishot network coding, see \cite[Introduction and Section $1$]{GeneralizedRS}. A sum-rank metric code $\Cc$ is an $\FF_q$-linear subspace of the product $\prod_{i=1}^t{\rm M}_{m_i,n_i}(\FF_q)$, endowed with the sum-rank metric. This metric comes from the weight wt$_{SR}$ defined as wt$_{SR}(C) = \sum_{i=1}^trank(C_i)$, where $C=(C_1,\dots ,C_t)\in \Cc \subset \prod_{i=1}^t{\rm M}_{m_i,n_i}(\FF_q)$. For this weight, MacWilliams identities were proven in \cite[Section $5$]{SRC_Byrne} and for Maximum rank-sum distance (MRSD) codes, it was also proven that their sum-rank weight distribution depends only on the size of the code \cite[Section $6$]{SRC_Byrne}.

Later, those three metrics were generalized, in the sense that the weights are defined on \textbf{subspaces} of the code. First definition of generalized Hamming weights (GHW) was given in \cite{WeiGHW} ; some basic properties are also demonstrated and it is explained that generalized weights characterize the code performance in the context of wire-tap channel of type II. As for the non-generalized weights, the same questions on the distributions of these generalized weights can be asked but in a slightly different way : we want to count the number of \textbf{subspaces} with a fixed generalized weight. In the case of the GHW, this question is completely answered in \cite{JurPelGeneralized}. Indeed, Jurrius and Pellikaan gave a closed-form formula for the generalized Hamming weight enumerator in \cite[Section 2, Theorem 1]{JurPelGeneralized} and they also demonstrated MacWilliams identities for GHW, see \cite[Section 7, Theorems 11 and 12]{JurPelGeneralized}. Moreover, they proved that the generalized Hamming weight distribution of an MDS code depends only on the parameters of the code. However, for the generalized sum-rank weight (GSRW), defined in \cite[Section $5$]{GeneralizedRS}, the same questions are still not answered. 

We summarize all of the existing results regarding these metrics and the MacWilliams theory in the three following Tables \ref{table1}, \ref{table2} and $\ref{table3}$.

\begin{table}[H]
\begin{center}
\begin{NiceTabular}{|c||c|c|}
\hline
\Block{2-1}{\diagbox{}{}} & \multicolumn{2}{ c |}{\textsc{Hamming metric}} \\
\cline{2-3}
& Hamming Weight & GHW  \\
 \hline
Weights enumerator & See \cite{Klove} &  See \cite{JurPelGeneralized}  \\
\hline
MacWilliams identities & See \cite{MacWiliams} & See \cite{JurPelGeneralized}  \\
\hline 
MDS codes & See \cite{Tolhuizen} & See \cite{JurPelGeneralized}  \\
\hline
\end{NiceTabular}
\vspace{0.2cm}
\caption{State-of-the-art about results on the Hamming weights distribution}
\label{table1}
\end{center}
\end{table}

\vspace{-0.4cm}

\begin{table}[h!]
\begin{center}
\begin{NiceTabular}{|c||c|c|}
\hline
\Block{2-1}{\diagbox{}{}} & \multicolumn{2}{ c |}{ \textsc{Rank metric}}  \\
\cline{2-3}
& Rank Weight & GRW  \\
 \hline
Weights enumerator & See  \cite{Shiromoto} & \textit{Not Answered} (\textit{N.A.})\\
\hline
MacWilliams identities &  See \cite{GadouleauYan,Ravagnani,Shiromoto} & \textit{N.A.} \\
\hline 
MRD codes &  See \cite{GadouleauYan,Ravagnani} & \textit{N.A.} \\
\hline
\end{NiceTabular}
\vspace{0.2cm}
\caption{State-of-the-art about results on the Rank weights distribution}
\label{table2}
\end{center}
\end{table}

\vspace{-0.4cm}

\begin{table}[h!]
\begin{center}
\begin{NiceTabular}{|c||c|c|}
\hline
\Block{2-1}{\diagbox{}{}} & \multicolumn{2}{ c |}{ \textsc{Sum-Rank metric}}  \\
\cline{2-3}
& Sum-Rank Weight & GSRW \\
 \hline
Weights enumerator & \textit{N.A.} & \textit{N.A.} \\
\hline
MacWilliams identities & See  \cite{GeneralizedRS} & \textit{N.A.} \\
\hline 
MRSD codes & See \cite{GeneralizedRS} & \textit{N.A.} \\
\hline
\end{NiceTabular}
\vspace{0.2cm}
\caption{State-of-the-art about results on the Sum-Rank weights distribution}
\label{table3}
\end{center}
\end{table}

The questions left in the column of the generalized rank weight in Table \ref{table2} will be at the core of this paper and will be completely answered. It is organized as follows. In Section \ref{Section2}, we define all the concepts and notation we use in this paper about the generalized rank weights and all the tools related to the weight distribution. Then, in Section \ref{Section3}, we obtain a MacWilliams identity for  the generalized rank weight distribution. Later, in Section \ref{Section4}, a closed-form formula is obtained for the generalized rank weight enumerator, tailoring the steps followed in \cite{JurPelGeneralized} to our metric. Finally, in Section \ref{Section5}, we solve the question on the generalized rank weight  distribution of MRD codes by showing that it depends only on the parameters of the code and not on the code itself. 

All the results of this paper are supported by examples and the SAGE \cite{SAGE} codes used to perform the computations can be found at : \url{https://github.com/JulienMLN/MacWilliamsIdentitiesForGeneralizedRankWeights}

\section{Definitions}\label{Section2}

Throughout this paper, an $[n,k]$ \textit{linear code} is nothing but a subspace of $\FF_{q^m}^n$ of dimension $k$ over $\FF_{q^m}$, with $m\geq 1$ an integer and $\FF_t$ the finite field with $t$ elements, $t$ a prime power. The number $k$ is the dimension of the code and $n$ is called its length. Such codes are said to be codes \textit{à la} Gabidulin as opposed to Block codes or Delsarte codes where codes are subspaces of rectangular matrices with coefficients in $\FF_q$. This later approach is used by Shiromoto in \cite{Shiromoto} and by Ravagnani in \cite{Ravagnani}. In this same paper, Section $3$, Ravagnani provides a complete comparison between Delsarte codes and Gabidulin codes in terms of their weight distribution, dimensions and duality theory.

\subsection{Generalized rank metric}

First definitions of generalized rank weights (GRW) were given over finite fields extension $\FF_{q^m}/\FF_q$. In \cite{BerhuyFaselGarotta}, it is shown that same definitions can be extended over any finite extension $\LL/\KK$. For the purpose of this paper, we will work over finite fields since our goal is to enumerate codewords and subspaces of a code.

\begin{definition}[see \cite{JurriusPellikaan}]\label{DefJurriusPelik}
    Let $\Cc$ be an $\FF_{q^m}$-linear code with parameters $[n,k]$, that is a code of length $n$ and dimension $k$.

    Let us pick a $\FF_q$-basis  $\mathcal{B}$ of $\FF_{q^m}$. For any vector $\cv = (c_1,\ldots,c_n)\in \FF_{q^m}^n$, let  $M_{\mathcal{B}}(\cv)\in {\rm M}_{m\times n}(\FF_q)$ be the matrix whose entries are the coordinates of each $c_i$ with respect to $\mathcal{B}$. The {\it rank support} of $\cv$, denoted by Rsupp$(\cv)$, is the $\FF_q$-linear row space of $M_{\mathcal{B}}(\cv)$.  We define wt$_R(\cv)$ to be the dimension of Rsupp$(\cv)$, that is, the rank of $M_{\mathcal{B}}(\cv)$. This does not depend on the choice of $\mathcal{B}$. 

    Let $\mathcal{D}$ be an $\FF_{q^m}$-linear subspace of $\Cc$. Then, Rsupp$(\mathcal{D})$, the {\it rank support} of $\mathcal{D}$, is the $\FF_q$-linear subspace of $\FF_q^n$ generated by Rsupp$(\textbf{d})$, for all $\textbf{d}\in \mathcal{D}$. Then, wt$_R(\mathcal{D})$ is defined as the dimension of Rsupp$(\mathcal{D})$.

    Finally, for $1\leq r \leq k$, the {\it $r$-$th$ generalized rank weight} of the code $\Cc$, denoted by $M_r(\Cc)$, is defined as $$M_r(\Cc) \defegal  \underset{\substack{\mathcal{\mathcal{D}} \subset \Cc \\ \dim (\mathcal{\mathcal{D}}) = r}}{\min} \text{wt}_R(\mathcal{\mathcal{D}}).$$ \vspace{0.25cm}
\end{definition}

\begin{remark}\label{remarkcoord}
    For a codeword $\cv = (c_1,...,c_n)\in \Cc$, wt$_R(\cv)$ is nothing but the dimension of the $\FF_q$-linear subspace generated by the coordinates of $\cv$. In other words, wt$_R(\cv) = \dim_{\FF_q} \text{Span}(c_1,...,c_n)$.
\end{remark}

There exist three more definitions of the generalized rank weights. All these definitions are equivalent under the assumption that $m\geq n$, with $m$ the degree of the field extension and $n$ the length of the code. The proofs can be found in \cite{JurriusPellikaan} for codes defined over finite fields and in  \cite{BerhuyFaselGarotta} for codes defined over any finite extension. The assumption $m\geq n$ will be assumed throughout this paper so that we can use any definition independently.

\begin{definition}[See \cite{OggierSboui}]
    Let $\Cc$ be an $\FF_{q^m}$-linear code with parameters $[n,k]$.
    The $r$-th generalized rank weight is given by : 
    $$M_r(\Cc) = \underset{\substack{\mathcal{\mathcal{D}} \subset \Cc \\ \dim (\mathcal{\mathcal{D}}) = r}}{\min} \text{maxwt}_R(\mathcal{D}), $$ where maxwt$_R(\mathcal{D})$ is $\underset{d\in \mathcal{D}}{\max}$ wt$_R(\textbf{d})$.
\end{definition}

\begin{remark}\label{MinMaxRW}
   From this definition, the first and last GRW can be rewritten as follows :  \[ \displaystyle M_1(\Cc) = \min_{\substack{\cv \in \Cc \\ \cv \neq 0}}\text{wt}_R(\cv) \hspace{1cm}\text{and} \hspace{1cm}\displaystyle M_k(\Cc) = \max_{\cv \in \Cc}\text{wt}_R(\cv), \]for any $[n,k]$ linear code $\Cc$.
\end{remark}

\begin{definition}[See \cite{KuriMatsuUye}]\label{DefKuriMatsuUye}
    Let $\Cc$ be an $\FF_{q^m}$-linear code with parameters $[n,k]$.
    The $r$-th generalized rank weight is given by : 
    $$M_r(\Cc) = \underset{\substack{V \subset \FF_{q^m}^n, \ V^*=V \\ \dim (\Cc \cap V) \geq r}}{\min} \dim V, $$
    where $V^*$ is the Galois closure of $V$, which is the smallest subspace of $\FF_{q^m}^n$ that contains $V$ and that is closed under the component-wise action of the Frobenius map of the extension $\FF_{q^m}/\FF_q$.
\end{definition}

There is another definition of generalized rank weights using Galois closure in \cite{JeromeDucoat}. We will not be using this definition ; for more details on Galois closure, one may read \cite{JurriusPellikaan, BerhuyFaselGarotta}. 

\begin{definition}
    For an $[n,k]$-linear code $\Cc$, the collection of weights $M_1(\Cc), M_2(\Cc),$ $ \ldots ,  M_k(\Cc)$ is called the \textit{generalized rank weight hierarchy of $\Cc$}. 
    In particular, for $r=1$, $M_1(\Cc)$ is called the \textit{minimum rank distance/weight}.
\end{definition}

Some well-known properties of the rank weight hierarchy are summarized in the following proposition.

\begin{proposition}\label{PropPoidsgeneralites}
Let $\Cc$ be an $\FF_{q^m}$-linear code with parameters $[n,k]$.
    \begin{enumerate}
        \item The rank weight hierarchy is  increasing (\cite[Lemma $9$]{KuriMatsuUye}) : $$1\leq M_1(\Cc) < M_2(\Cc) < \cdots < M_k(\Cc) \leq n.$$ 
        \item For $1\leq r \leq k$, we have a generalized Singleton bound (\cite[Corollary $15$]{KuriMatsuUye}) : $$M_r(\Cc) \leq n - k +r.$$
        \item Let $\Cc^{\perp} = \{ \cv' \in \FF_{q^m}^n \ ; \ \langle \cv',\cv\rangle =0, \forall \cv\in \Cc \}$ be the dual code of $\Cc$, where $\langle \cdot , \cdot \rangle$ is the standard inner product over $\FF_{q^m}^n$. Then (\cite[Theorem I.3]{JeromeDucoat}) : $$\{ M_r(\Cc) \ ; \ 1\leq r \leq k\} = \{ 1, \ldots , n \}\setminus \{ n+1-M_r(\Cc^{\perp}) \ ; \ 1\leq r \leq n-k\}.$$ 
    \end{enumerate}
\end{proposition}
\vspace{0.25cm}

\begin{definition}\label{MRDdef}
    Let $\Cc$ be an $\FF_{q^m}$-linear code with parameters $[n,k]$. The code $\Cc$ is said to be $r$-MRD (Maximum Rank Distance) if $M_r(\Cc)=n-k+r$, that is, when the $r$-th rank distance reaches the generalized Singleton bound.
    In particular, for $r=1$, we say that $\Cc$ is an MRD code.
\end{definition}

\subsection{Generalized rank weight distribution and enumerator}

For the rank metric, the weight distribution of a code $\Cc$ specifies the number of codewords of each possible weight. Since, we work with the generalized rank metric, we have a slightly different definition since generalized weights are taken over subspaces of $\Cc$.

\begin{definition}
     The generalized rank weight distribution of a code $\Cc$ consists in the collection $(A^r_w(\Cc))_{r,w}$ where $$A_w^r(\Cc) \defegal \#\{ \Dc \subset \Cc \ \mid \ \dim_{\FF_{q^m}}\Dc = r \ ; \ \text{wt}_R(\Dc) = w\},$$
     for $r$ and $w$ nonnegative integers, and where $\#E$ is the cardinality of a set $E$.
\end{definition}

Now, we can state some easy properties of the numbers $A^r_w(\Cc)$. 

\begin{proposition}\label{PropDistribution}
    Let $\Cc\subset \FF_{q^m}^n$ be an $[n,k]$ code and $(A^r_w(\Cc))_{w,r}$ its GRW distribution. 
    \begin{enumerate}
        \item If $r=0$, we have $A_0^0(\Cc) = 1$ and $A^0_w=0$, for all $w>0$.
        \item For all $w\geq 0$ and $r>k$, we have $A^r_w(\Cc) =0$.
        \item For all $w<r$ or $w>n$, we have $A^r_w(\Cc) =0$.
        \item We have $\displaystyle A^1_w(\Cc) =\frac{\#\{ \cv \in \Cc \ \mid \ \text{\normalfont wt}_R(\cv)=w \}}{q^m-1}$.
    \end{enumerate}
    \begin{proof}
        Items $1$ and $2$ come directly from the definition of $A_w^r(\Cc)$. For item $3$, since the sequence $(M_r(\Cc))_r$ is strictly increasing, then necessarily $w\geq r$, and by Remark \ref{MinMaxRW}, necessarily $w\leq n$ . For the last item, we used the fact that two $\FF_{q^m}$-proportional codewords have the same rank weight, see for instance  \cite[Proposition 2.2]{BerhuyFaselGarotta}. Thus all nonzero codewords in a $1$-dimensional subspace of $\Cc$ have the same rank weight.
    \end{proof}
\end{proposition}

Thank to this proposition, we can constrain the computation of the GRW distribution $(A^r_w(\Cc))_{w,r}$ to the computation of a finite number of the $A^r_w(\Cc)$. More precisely, we only have to compute $A^r_w(\Cc)$ for $r\in \llbracket 1,k\rrbracket$ and $w\in \llbracket r,n\rrbracket$.

\begin{definition}
    The generalized rank weight enumerator of an $[n,k]$ code are the homogeneous polynomials counting for each dimension $r$ the number $A^r_w(\Cc)$ :
\begin{eqnarray*}
W^{r}_{\mathcal{C}}(X,Y) \defegal \sum_{w=0}^n A^{r}_w(\mathcal{C})X^wY^{n-w},\\
\end{eqnarray*}  
for $r=0,\ldots,k$.
\end{definition}

\subsection{Gaussian binomial coefficient}

\begin{definition}
    For $q$ a prime power and $n,k$ two nonnegative integers, we define the \textit{Gaussian binomial coefficients} or \textit{q-binomial coefficients} to be the following integers : $$\begin{bmatrix}
        n\\ k
    \end{bmatrix}_q \defegal \left\{  \begin{array}{cl}
       0  & \text{if } k>n   \\
       1  & \text{if } k=0 \text{ or } n=0\\
       \dfrac{(1-q^n)(1-q^{n-1})\cdots (1-q^{n-k+1})}{(1-q^k)(1-q^{k-1})\cdots (1-q)} & \text{else}
    \end{array}\right.$$
\end{definition}

The relevant properties of these coefficients are summarized in the next lemma.

\begin{lemma}\label{lemmaq-Coeff}
For $a,b,c$ nonnegative integers, the following properties hold.
    \begin{enumerate}
        \item The number $\begin{bmatrix}
        a\\ b
    \end{bmatrix}_q$ is the number of $b$-dimensional subspaces of an $a$-dimensional vector space over $\FF_q$.
        \item Reflection property : $\displaystyle \begin{bmatrix}
            a\\b
        \end{bmatrix}_q = \begin{bmatrix}
            a\\a-b
        \end{bmatrix}_q$.
        \item We have $\begin{bmatrix}
            a\\b
        \end{bmatrix}_q\begin{bmatrix}
            b\\c
        \end{bmatrix}_q = \begin{bmatrix}
            a\\c
        \end{bmatrix}_q\begin{bmatrix}
            a-c\\a-b
        \end{bmatrix}_q$.
        \item We have a $q$-analogue formula of the binomial theorem :  $$\prod_{k=0}^{n-1}(x+q^ky) = \sum_{k=0}^nq^{\frac{k(k-1)}{2}}\begin{bmatrix}
            n\\k
        \end{bmatrix}_qx^{n-k}y^k, \ \ \ \ \forall x,y\in \mathbb{C}.$$
        \item We have a $q$-analogue of the Vandermonde identity : $$\begin{bmatrix}
            a+b\\c
        \end{bmatrix}_q=\sum_{l=0}^cq^{l(l+b-c)}\begin{bmatrix}
            a\\l
        \end{bmatrix}_q\begin{bmatrix}
            b\\c-l
        \end{bmatrix}_q.$$
    \end{enumerate}
\end{lemma}
\begin{proof}
    For items $1,2,3$, \cite[Chapter 1]{CombinatoricStanley} is a good introduction to $q$-binomial coefficients and their properties. Items 4 and 5 may be found, respectively, in \cite[Chapter 5]{QuantumKac} and \cite[Chapter 1, Correction of exercise 100]{CombinatoricStanley}.
\end{proof}

\section{MacWilliams identity on the Generalized Rank Weight distribution.}\label{Section3}

To obtain a MacWilliams identity for the generalized rank weight distribution, we draw from the framework proposed by Ravagnani in \cite[Section 4]{Ravagnani}.

\begin{notation}
For the beginning of this section, we work with any finite extension $\LL/\KK$. We will indicate to the reader when we require the finiteness of the fields.
\end{notation}

\begin{notation}
    For a subspace $U\subset \KK^n$, we denote by $\Ul\subset \LL^n$ the $\LL$-subspace spanned by all the elements of $U$. 
\end{notation}

\begin{lemma}\label{lemmeU_L}
    Let $U$ be a subspace of $\KK$. The following hold : 
    \begin{enumerate}
        \item $\dim_\KK U = \dim_\LL \Ul$.
        \item $(U^{\perp_K})_\LL = ( \Ul)^{\perp_L}$, where $E^{\perp_K}$ is the orthogonal of $E\subset \KK^n$ in $\KK^n$ and $F^{\perp_L}$ is the orthogonal of $F\subset \LL^n$ in $\LL^n$
    \end{enumerate}  
\end{lemma}
\begin{proof}
    For item $1$, one may see \cite[Lemma 3.6]{BerhuyMolina}.

    For item $2$, using the first item, we first have \begin{eqnarray*}
        \dim_\LL (\Ul)^{\perp_L} & =  &\dim_\LL\LL^n- \dim_\LL \Ul \\ & =& \dim_\LL\LL^n- \dim_\KK U \\ & =& \dim_\KK\KK^n - \dim_\KK U \\ & = &\dim_\KK U^{\perp_K} \\ & =& \dim_\LL (U^{\perp_K})_\LL
    \end{eqnarray*}
    Now, to obtain the equality, it suffices to prove the inclusion $(U^{\perp_K})_\LL \subset (\Ul)^{\perp_L}$.
    Let us take $x\in (U^{\perp_K})_\LL$ and $y\in\Ul$. We can write $x=\sum_i\lambda_i\tilde{u_i}$, with $\lambda_i\in \LL$ and $\tilde{u_i}\in U^{\perp_K}$ and also $y=\sum_j\mu_ju_j$ with $\mu_j\in \LL$ and $u_j\in U$.
    We then compute : 
    \begin{eqnarray*}
        \langle x,y\rangle_\LL  =  \langle \sum_i\lambda_i\tilde{u_i},\sum_j\mu_ju_j \rangle_\LL =  \sum_i\lambda_i\sum_j\mu_j \langle\tilde{u_i},u_j\rangle_\LL =  \sum_i\lambda_i\sum_j\mu_j \langle\tilde{u_i},u_j\rangle_\KK = 0
    \end{eqnarray*}
    Last two equalities are obtained by definition of $\tilde{u_i}$ and $u_j$. Finally, we have $x\in (\Ul)^{\perp_\LL}$
\end{proof}

Thanks to this lemma, for the rest of the paper, we will write $A^\perp$ to denote the orthogonal complement of a subspace A, whether in $\LL$ or in $\KK$, to ease notation.

\begin{notation}\label{DefCU}
    For a subspace $\Cc$ of $\LL^n$ and a subspace $U$ of $\KK^n$, we define the set \[ \Cc(U) \defegal \{ \cv \in \Cc \ \mid \ \text{Rsupp}(\cv) \subset U^\perp \}\]. 
\end{notation}

\begin{lemma}\label{caracterisatinC(U)}
    Let $U$ be a subspace of $\KK^n$ and $\Cc$ be a subspace of $\LL^n$.
    \begin{enumerate}
        \item For all $\cv \in \LL^n$, we have Rsupp($\cv) \subset U$ if and only if $\cv \in \Ul$.
        \item We have the following equality : $\Cc(U) = \Cc\cap (\Ul)^\perp$. In particular, $\Cc(U)$ is an $\LL$-vector space.
        \item For all subspaces $\Dc \subset \Cc \subset \LL^n$, we have $\Dc \subset \Cc(U)$ if and only if Rsupp$(\Dc) \subset U^\perp$.
    \end{enumerate}
\end{lemma}

\begin{proof}

    Let $(\alpha_i)$ be a $\KK$-basis of $\LL$ and take $\cv=(c_1,...,c_n)\in \LL^n$ such that Rsupp($\cv$)$\subset U$. We can write for all $i\in \llbracket 1,n\rrbracket$, $c_i = \sum_{j=1}^mc_{ji}\alpha_j$n with $c_{ij}\in \KK$. Then Rsupp$(\cv) = \text{Span}_\KK(v_1,...,v_m)$, where $v_i \defegal (c_{1i},\dots,c_{ni})$. By assumption $v_i\in U$. Finally, we have $\cv=\sum_{s=1}^m\alpha_sv_s \in \Ul$
    Conversely, assume that $\cv\in \Ul$, in other words, $\cv = \sum_{i=1}^t\lambda_iu_i$, for some $t$ positive integer, $\lambda_i \in \LL$ and $u_i\in U$. As previously, we write $\lambda_i = \sum_{j=1}^m\mu_{ji}\alpha_j$. We then have $\cv = \sum_{i=1}^t\lambda_iu_i = \sum_{j=1}^m \left( \sum_{i=1}^t \mu_{ji}u_i\right)\alpha_j$. In this case, we finally get Rsupp$(\cv) = \text{Span}_\KK(\sum_{i=1}^t \mu_{1i}u_i, \dots , \sum_{i=1}^t \mu_{mi}u_i)\subset U$, since $u_i\in U$. The second point is then clear. 

    \noindent For the last point, if $\Dc \subset \Cc(U)$ then Rsupp($\textbf{d})\in U^\perp$ for all $\textbf{d} \in \Dc$. Since $U^\perp$ is a subspace and by definition of Rsupp$(\Dc)$ (see Definition \ref{DefJurriusPelik}), Rsupp$(\Dc) \subset U^\perp$. The converse is obtained by a similar line of reasoning.
\end{proof}

\begin{proposition}\label{DimC(U)}
    We have $\dim_\LL\Cc(U) = \left\{ \begin{array}{rl}
       0  & \text{if } \dim_\KK U >n-M_1(\Cc)   \\
       \dim_\LL \Cc - \dim_\KK U  & \text{if } \dim_\KK U < M_1(\Cc^\perp) 
    \end{array}\right.$
\end{proposition}

\begin{proof}
    Let $t\defegal \dim_\KK U$ and assume that $t> n-M_1(\Cc)$. Let $\cv \in \mathcal{C}(U)$. By definition, $U$ is contained in $\text{Rsupp}(\cv)^{\perp}$ and thus $t\leq n-\text{wt}_R(\cv)$. It follows that wt$_R(\cv) \leq n-t < M_1(\Cc)$ and $\cv$ is necessarily the zero vector by definition of the first rank weight. In such case, $\dim_\LL \Cc(U)=0$.

    Next, assume that $t < M_1(\Cc^\perp)$. Definition \ref{DefKuriMatsuUye} applied to $\Cc^\perp$ yields $\dim_\LL \Ul \cap \Cc^\perp = 0$. We have $\Cc(U)^\perp = (\Cc\cap \Ul^\perp)^\perp = \Cc^\perp + \Ul$ and thus $$\dim_\LL \Cc(U)^\perp= \dim_\LL \Cc^\perp + \dim_\LL \Ul - \dim_\LL \Cc^\perp\cap \Ul = \dim_\LL \Cc^\perp + \dim_\LL \Ul.$$ By Lemma \ref{lemmeU_L} and since $\dim_\LL \Cc(U)^\perp=n-\dim_\LL\Cc(U)$ and $\dim_\LL \Cc^\perp=n-\dim_\LL \Cc$, we get the result.
\end{proof}

\begin{remark}\label{Rmq_2cas}
    The formulas in Proposition \ref{DimC(U)} are consistent in the sense that if both conditions on $\dim_{\KK}U$ hold, then the formulas agree. Indeed, necessarily $t=k$, with $k$ the dimension of $\Cc$ and $t$ the dimension of $U$. Indeed, the Singleton bound for $\Cc$ gives $M_1(\Cc) \leq n-k+1$, that is $k \leq n-M_1(\Cc)+1$. Similarly, $M_1(\Cc^\perp) \leq k+1$. Using these inequalities, if $t>n-M_1(\Cc)$, then $t \geq n-M_1(\Cc) +1\geq k$. Else, if $t< M_1(\Cc^\perp)$, then $t\leq M_1(\Cc^\perp) -1\leq k$. But finally, in such situation, both cases in Proposition $\ref{DimC(U)}$ coincide. 
\end{remark}

\begin{example}\label{Exemple1}
    Consider the $[3,1]$-cyclic code\footnote{For an introduction to cyclic codes, one may read \cite[Chapter 7]{MacWSloaneTheory} or \cite[Chapter 4]{HuffmanPless}} $\Cc \subset \FF_{2^4}[x]/(x^3-1)$ generated by $g(x)=(x+1)(x+\alpha^2+\alpha)$, with $\alpha$ such that $\FF_{2^4} = \FF_2(\alpha)$ and $\alpha^4 = \alpha+1$. We know that $M_1(\Cc)$ can be $1,2$ or $3$. Since no cyclic code is MRD (see \cite[Proposition $37$]{DucOggier1}), then $M_1(\Cc) \neq 3$. Also $M_1(\Cc)\neq 1$ using \cite[Theorem 4.8]{BerhuyMolina}, thus we have $M_1(\Cc) =2$. Moreover, $\Cc^\perp$ is a $[3,2]$ cyclic code generated by $h(x)=x+\alpha^2+\alpha+1$ (see \cite[Proposition 4.17]{BerhuyMolina}). Proposition \ref{PropPoidsgeneralites} $(3)$ yields $M_1(\Cc^\perp)=1$ and $M_2(\Cc^\perp)=3$.

    \noindent Now,  $\dim_{\FF_{2^4}}\Cc(U) = \left\{ \begin{array}{rl}
       0  & \text{if } \dim_{\FF_2} U \geq 2   \\
       1   & \text{if } \dim_{\FF_2} U =0 
    \end{array}\right.$.
    As one may notice, Proposition \ref{DimC(U)} does not give a value for  $\dim_{\FF_{2^4}}\Cc(U)$ when $\dim_{\FF_2} U =1$, so we have to compute it by hand. Let us take $U\subset \FF_2^3$ with $\dim_{\FF_2}U=1$ and write $U=\text{Span}_{\FF_2}\{ (u_1,u_2,u_3)\}$. Since $\Cc = \text{Span}_{\FF_{2^4}}\{ (\alpha^2+\alpha, \alpha^2+\alpha +1, 1)\}$, we get that $$\Cc(U) = \Cc\cap (U^\perp)_{\FF_{2^4}} = \left\{ \begin{array}{rl}
       \Cc  & \text{if } U = \FF_2\cdot (1,1,1)   \\
       \{0\}   & \text{else} 
    \end{array}\right.$$ Finally, $\dim_{\FF_{2^4}}\Cc(U) = \left\{ \begin{array}{rl}
       0  & \text{if } \dim_{\FF_2} U \geq 2 \text{ or } (\dim_{\FF_2} U = 1 \text{ and } U \neq \FF_2\cdot (1,1,1)),  \\
       1  & \text{if } \dim_{\FF_2} U = 0 \text{ or } U = \FF_2\cdot (1,1,1).  \\ 
    \end{array}\right.$
\end{example}

\begin{remark}
    As illustrated by the previous example, the formula given in Proposition \ref{DimC(U)} does not always yield an answer for the dimension of $\Cc(U)$ when the two conditions are disjoint. In such cases, one must determine the dimension of $\Cc(U)$ directly.
\end{remark}

In order to get a MacWilliams identity, we first have to relate $\Cc(U)$ and $\Cc^\perp(U)$.

\begin{lemma}\label{liendimRavagna}
    Let $\Cc$ be an $[n,k]$ code over $\LL$. For any subspace $U\subset \KK$, the following equality holds : 
    \[ \dim_\LL \Cc(U) = \dim_\LL\Cc - \dim_\KK U + \dim_\LL \Cc^\perp(U^\perp) \]
\end{lemma}
\begin{proof}
    By Lemma \ref{caracterisatinC(U)}, we have the equality $\Cc(U)^\perp          = (\Cc \cap U^\perp)^\perp = \Cc^\perp + U$.

    Taking the dimension over $\LL$ on both sides of the equality and using the Grassmann formula, we get the result since $\Cc^\perp \cap U = \Cc^\perp(U^\perp)$.
\end{proof}

From now on, we will assume that $\LL = \FF_{q^m}$ and $\KK=\FF_q$.

\begin{lemma}\label{lemmalienA}
    For $\Cc \subset \LL^n, r\leq \dim \Cc$ and $t\leq n$, we have the following  formula : 

    \[ \sum_{\substack{U\subset \KK^n \\ \dim_\KK U = t}}\begin{bmatrix}
        \dim_\LL \Cc(U) \\ r
    \end{bmatrix}_{q^m} = \sum_{w=0}^n\begin{bmatrix}
        n-w\\t
    \end{bmatrix}_qA^r_w(\Cc)\]
\end{lemma}
\begin{proof}
    This proof is based on counting the elements of the following set in two different ways : $$X\defegal \{ (\Dc,U) \ ; \ U \subset \FF_q^n, \dim_{\FF_q} U =t, \Dc\subset \Cc(U), \dim_\LL \Dc =r\}$$

    On the one hand, for a fixed $U$ with $\dim U =t$, the number of subspace $\Dc$ such that $(\Dc,U) \in X$ is $\begin{bmatrix}
        \dim_\LL \Cc(U) \\ r
    \end{bmatrix}_{q^m}$. So,  $\#X = \displaystyle\sum_{\substack{U\subset \KK^n \\ t=\dim_\KK U}} \begin{bmatrix}
        \dim_\LL\Cc(U) \\r
    \end{bmatrix}_{q^m}$.
    
    On the other hand, there are $A_w^r(\Cc)$ subspaces $\Dc$ of $\Cc$ with dimension $r$ and $\text{wt}_R(\Dc)=w$. Finding a subspace $U$ such that $\Dc \subset \Cc(U)$ can be rewritten as Rsupp$(\Dc) \subset U^\perp$ by Lemma \ref{caracterisatinC(U)}. Since $\dim_\KK \text{Rsupp}(\Dc) = \text{wt}_R(\Dc) =w$, the number of such $U$ is then $\begin{bmatrix}
        n-w\\t
    \end{bmatrix}_q$ by Lemma \ref{lemmaq-Coeff}. Finally, summing over all the possible weights $w$, we get the result.
\end{proof}

\begin{example}\label{Exemple2}
    Let us keep the code $\Cc$ defined in Example \ref{Exemple1}.  For $r=0$ and $r\geq 2$, the GRW distribution is completely determined by Proposition \ref{PropDistribution}. Then, for $r=1$, we use the previous Lemma \ref{lemmalienA} and we get several equations for different choices of the parameter $t$ : 
    \begin{enumerate}[label=$\bullet$]
        \item For $t=0$ : $\begin{bmatrix}
            1\\1
        \end{bmatrix}_{2^4} = A^1_0(\Cc) + A^1_1(\Cc) + A^1_2(\Cc) + A^1_3(\Cc)$
        \item For $t=1$ : $\begin{bmatrix}
            1\\1
        \end{bmatrix}_{2^4} = \begin{bmatrix}
            3\\1
        \end{bmatrix}_2A^1_0(\Cc) + \begin{bmatrix}
            2\\1
        \end{bmatrix}_2A^1_1(\Cc) + \begin{bmatrix}
            1\\1
        \end{bmatrix}_2A^1_2(\Cc)$
        \item For $t=2$ : $0 = \begin{bmatrix}
            3\\2
        \end{bmatrix}_2A^1_0(\Cc) + \begin{bmatrix}
            2\\2
        \end{bmatrix}_2A^1_1(\Cc)$
        \item  For $t=3$ : $0=\begin{bmatrix}
            3\\3
        \end{bmatrix}_2A^1_0(\Cc)$
    \end{enumerate}
    Solving those four equations, we find that $(A^1_0(\Cc),A^1_1(\Cc),A^1_2(\Cc),A^1_3(\Cc))= (0,0,1,0)$. 
\end{example}

From these two previous lemmas, we can derive a MacWilliams identity relating the generalized rank weight distributions of a code $\Cc$ and its dual $\Cc^\perp$. Sometimes, it is named as the \textit{moments} of MacWilliams identity.

\begin{theorem}\label{MacWIdentDistribution}
    Let $\Cc$ be an $[n,k]$ linear code over $\FF_{q^m}$. Let $\Cc^\perp$ be the dual code of $\Cc$.

    \noindent Then for all $0\leq t \leq n-k$ and all $r\geq 0$, we have the following formula : 

    \[ \sum_{w=0}^n\begin{bmatrix}
        n-w\\t
    \end{bmatrix}_qA^r_w(\Cc^\perp) = \sum_{p=0}^n\sum_{l=0}^rA_p^l(\Cc)q^{ml(l+n-k-t-r)}\begin{bmatrix}
        n-p\\n-t
    \end{bmatrix}_q\begin{bmatrix}
        n-k-t\\r-l
    \end{bmatrix}_{q^m} \]
\end{theorem}
\begin{proof}
    Apply consecutively Lemmas \ref{lemmalienA} and \ref{liendimRavagna} to $\Cc^\perp$ :
     \begin{align*}
        \sum_{w=0}^n\begin{bmatrix}
        n-w\\t
    \end{bmatrix}_qA^r_w(\Cc^\perp) & = \sum_{\substack{U\subset \KK^n \\ \dim_\KK U = t}}\begin{bmatrix}
        \dim_\LL \Cc^\perp(U) \\ r
    \end{bmatrix}_{q^m} \\ & = \sum_{\substack{U\subset \KK^n \\ \dim_\KK U = t}}\begin{bmatrix}
        \dim_\LL \Cc^\perp - \dim_\KK U + \dim_\LL \Cc(U^\perp) \\ r
    \end{bmatrix}_{q^m}
    \end{align*}
    \[   \]
    Since the map $U\longmapsto U^\perp$ is a bijection, we can rewrite the previous equality : 
    $$\sum_{w=0}^n\begin{bmatrix}
        n-w\\t
    \end{bmatrix}_qA^r_w(\Cc^\perp) = \sum_{\substack{V\subset \KK^n \\ \dim_\KK V = n-t}}\begin{bmatrix}
        \dim_\LL \Cc^\perp - \dim_\KK V^\perp + \dim_\KK \Cc(V) \\ r
    \end{bmatrix}_{q^m}$$
    Then, using Lemma \ref{lemmaq-Coeff} $(5)$ with $a=\dim_\KK \Cc(V)$, $b=\dim_\LL \Cc^\perp - \dim_\KK V^\perp$ and $c=r$, we get
   
    \[ \sum_{w=0}^n \begin{bmatrix}
        n-w\\t
    \end{bmatrix}_qA^r_w(\Cc^\perp) = \sum_{\substack{V\subset \KK^n \\ \dim_\KK V = n-t}}\sum_{l=0}^rq^{ml(l+n-k-t-r)}\begin{bmatrix}
        \dim_\LL\Cc(V)\\ l
    \end{bmatrix}_{q^m}\begin{bmatrix}
        n-k-t \\r-l
    \end{bmatrix}_{q^m}\]
    Finally, apply again Lemma \ref{lemmalienA} to get the formula.
\end{proof}

\begin{example}\label{ex3.13}
    Keeping the code $\Cc$ defined in the previous Example \ref{Exemple2}, we can compute the GRW distribution of the dual code $\Cc^\perp$ whose parameters are $[3,2]$. Using the equations given by the previous Theorem \ref{MacWIdentDistribution}, we obtain \[ (A_0^1(\Cc^\perp),A_1^1(\Cc^\perp),A_2^1(\Cc^\perp),A_3^1(\Cc^\perp),A_2^2(\Cc^\perp),A_2^3(\Cc^\perp))= (0,1,4,12,0,1).\]
    For instance, let us compute the distribution for $r=1$. We get three different equations : 
    \begin{itemize}
        \item For $t=0$ : $A^1_0(\Cc^\perp) + A^1_1(\Cc^\perp) + A^1_2(\Cc^\perp) + A^1_3(\Cc^\perp) = \begin{bmatrix}
            2\\1
        \end{bmatrix}_{2^4} = 17$
        \item For $t=1$ : $7A^1_0(\Cc^\perp) + 3A_1^1(\Cc^\perp) + A_2^1(\Cc^\perp) = \begin{bmatrix}
            3\\2
        \end{bmatrix}_{2^4}= 7$
        \item For $t=2$ : $7A_0^1(\Cc^\perp)+A^1_1(\Cc^\perp) = 1$
    \end{itemize}
    Since $A_0^1(\Cc^\perp)=0$, we easily solve this system and get the result.
\end{example}

\section{Formula for the Generalized Rank Weight enumerator}\label{Section4}

Let us reconsider Lemma \ref{lemmalienA} and define some notation for later results.

\begin{definition}
Let $r$ be a nonnegative integer. For a subspace $\Cc$ of $\LL^n$ and a subspace $U$ of $\KK^n$, we set 
 \[ B_U^r(\Cc) \defegal \# \{ \Dc \subset \Cc(U) \ ; \ \dim_\LL \Dc =r \} \] \[ B_t^r(\Cc) \defegal \displaystyle\sum_{\substack{U\subset \KK^n \\ \dim U = t}} B_U^r(\Cc) \].
\end{definition}

\noindent By Lemma \ref{lemmaq-Coeff} (1), we easily get that $B_U^r(\Cc)= \begin{bmatrix}
    \dim_\LL \Cc(U) \\ r
\end{bmatrix}_{q^m}$.

\begin{proposition}\label{FormuleBrt}
    We have $B_t^r(\Cc) = \left\{ \begin{array}{cl}
       0  & \text{if } t >n-M_r(\Cc)   \\
       \begin{bmatrix}
           n\\ t
       \end{bmatrix}_q\begin{bmatrix}
           \dim_\LL \Cc - t \\ r
       \end{bmatrix}_{q^m}  & \text{if } t < M_1(\Cc^\perp) 
    \end{array}\right.$
\end{proposition}

\begin{proof}
    The second case is a direct corollary of Proposition \ref{DimC(U)} since there are exactly $\begin{bmatrix}
           n\\t
       \end{bmatrix}_q$ subspaces of $\KK^n$ with fixed dimension $t$.

    Then, assume that $t > n -M_r(\Cc)$ and take an $r$-dimensional subspace $\Dc \subset \Cc(U)$. By Lemma \ref{caracterisatinC(U)}, this condition becomes Rsupp$(\Dc) \subset U^\perp$. Then, taking the dimension on both sides,  wt$_R(\Dc) \leq n-t$. Using the hypothesis on $t$, we have wt$_R(\Dc) \leq n-t < M_r(\Cc)$, which, by definition of $M_r(\Cc)$, leads to the non-existence of such a subspace $\Dc$. Then $B_U^r =0$ and the result follows.
\end{proof}

\begin{remark}
    Generalized Singleton bounds for $\Cc$ and $\Cc^\perp$ give $M_r(\Cc^\perp) \leq k+r$ and $k-r\leq n-M_r(\Cc)$, where $k\defegal \dim\Cc$. Using these inequalities, we first have, if $t> n-M_r(\Cc)$, then $t\geq k-r+1$ that is $r-1\geq k-t$. And secondly, if $t< M_1(\Cc^\perp)$, then $t<M_1(\Cc^\perp)<M_r(\Cc^\perp) \leq k+r$, and so $-r+1\leq k-t$. One may say that those two conditions might overlap for $k-t \in \llbracket -r+1 ; r-1\rrbracket$. But for such values, the two formulae coincide by definition of the Gaussian binomial coefficient, and so the proposition \ref{FormuleBrt} makes sense. 
\end{remark}

\begin{proposition}\label{LienB->A}
    We have the following formula : $$B_t^r(\Cc) = \sum_{w=0}^n\begin{bmatrix}
        n-w\\ t
    \end{bmatrix}_qA_w^r(\Cc).$$
\end{proposition}

\begin{proof}
    This is just a rewritting of Lemma \ref{lemmalienA} with the new notation.
\end{proof}

The next lemma is a combinatorial result. A deeper understanding of this result can be found in the study of posets and Möbius inversion formulae (see \cite[Chapter 3, Section 3.10]{CombinatoricStanley}) and also in the use of recursion and inversion with counting functions (see \cite[Chapter III, Section 2.]{CombinatorialAigner}). Here is a simpler proof. 

\begin{lemma}\label{lemmaInversionGauss}
Let $E$ be a vector space of dimension $n+1$, with $n$ an integer. Let $\textbf{a}=(a_0,a_1,...,a_n)$ and $\textbf{b}=(b_0,b_1,...,b_n)$ two vectors of $E$. We have a Gaussian inversion type formula : 
 $$\forall j \in \llbracket 1,n \rrbracket, \ b_j = \sum_{i=0}^n\begin{bmatrix}
    i\\j
\end{bmatrix}_qa_i \iff \forall j \in \llbracket 1,n \rrbracket, \ a_j = \sum_{i=j}^n(-1)^{i-j}q^{\frac{(i-j)(i-j-1)}{2}}\begin{bmatrix}
    i\\j
\end{bmatrix}_qb_i$$
\end{lemma}
\begin{proof}
    Let us show that the two following matrices $\left( \begin{bmatrix}
        i\\j
    \end{bmatrix}_q\right)_{i,j=0,...,n}$ and 
    
    $\textstyle \left( (-1)^{i-j}q^{\frac{(i-j)(i-j-1)}{2}}\begin{bmatrix}
        i\\j
    \end{bmatrix}_q\right)_{i,j=0,...,n}$ are mutually inverse. 
    
    We have : 
    \begin{align*}
        \sum_{k=0}^n\begin{bmatrix}
            i\\k
        \end{bmatrix}_q (-1)^{k-j}q^{\frac{(k-j)(k-j-1)}{2}}\begin{bmatrix}
            k\\j
        \end{bmatrix}_q 
        = & \sum_{k=j}^i(-1)^{k-j}q^{\frac{(k-j)(k-j-1)}{2}}\begin{bmatrix}
            i\\j
        \end{bmatrix}_q\begin{bmatrix}
            i-j\\k-j
        \end{bmatrix}_q \\    = & \begin{bmatrix}
            i\\j
        \end{bmatrix}_q\prod_{p=0}^{i-j-1}(1-q^p)
    \end{align*}
    Indeed, the two equalities are obtained using item $(3)$ of Lemma \ref{lemmaq-Coeff} for the first one and then item $(4)$ for the second one.
If $i=j$, the computations give $\begin{bmatrix}
    i\\i
\end{bmatrix}\displaystyle\prod_{p=0}^{-1}(1-q^p) =1$, else the product is zero. So the matrix product is equal to the identity matrix.
\end{proof}

\begin{proposition}\label{AfonctionB}
    The following holds : $$A^r_w(\Cc) = \sum_{t=n-w}^n(-1)^{t+w-n}q^{\frac{(t+w-n)(t+w-n-1)}{2}}\begin{bmatrix}
        t \\ n-w
    \end{bmatrix}_qB_t^r(\Cc)$$
\end{proposition}

\begin{proof}
    Apply Lemma \ref{lemmaInversionGauss} with $\textbf{a} = (A_n^r,A_{n-1}^r,...,A_1^r,A_0^r)$ and 
    
    $\textbf{b} = (B_0^r(\Cc),B_1^r(\Cc),...,B_n^r(\Cc))$.
\end{proof}

We can now state our main theorem which gives a formula for the generalized rank weight enumerator : 

\begin{theorem}\label{THMEnumGenRank}
    The generalized rank weight enumerators for an $[n,k]$-code $\Cc$ defined over $\FF_{q^m}$ are given by $$W_\Cc^r(X,Y)  = \sum_{t=0}^nB_t^r(\Cc)X^{n-t}\prod_{p=0}^{t-1}(Y-q^pX),$$ for $r=0, \dots, k$.
\end{theorem}

\begin{proof} We have the following computations : 
\begin{eqnarray*} 
         W_\Cc^r(X,Y)  & = & \displaystyle \sum_{w=0}^n A^r_w(\Cc)X^wY^{n-w}  \notag \\
         &=& \displaystyle \sum_{w=0}^n\left( \sum_{t=n-w}^n(-1)^{t+w-n}q^{\frac{(t+w-n)(t+w-n-1)}{2}}\begin{bmatrix}
        t \\ n-w
    \end{bmatrix}_qB_t^r(\Cc) \right)X^wY^{n-w} \\
        &=& \displaystyle \sum_{t=0}^nB^r_t(\Cc)\left( \sum_{w=n-t}^n(-1)^{t+w-n}q^{\frac{(t+w-n)(t+w-n-1)}{2}}\begin{bmatrix}
        t \\ n-w
    \end{bmatrix}_qX^wY^{n-w}\right) \\
        &=& \displaystyle \sum_{t=0}^nB^r_t(\Cc)\left( \sum_{p=0}^t(-1)^pq^{\frac{p(p-1)}{2}}\begin{bmatrix}
        t \\ t-p
    \end{bmatrix}_qX^{p-t+n}Y^{t-p}\right) \\
       & =& \displaystyle \sum_{t=0}^nB^r_t(\Cc)\left( \sum_{p=0}^t(-1)^pq^{\frac{p(p-1)}{2}}\begin{bmatrix}
        t \\ p
    \end{bmatrix}_qX^{p}Y^{t-p}\right)X^{n-t} \\
    &=& \displaystyle \sum_{t=0}^nB_t^r(\Cc)\prod_{p=0}^{t-1}\left(Y-q^pX\right)X^{n-t}
\end{eqnarray*}
    The fourth equality comes from the change of index $p =t+w-n$, the fifth one is obtained by Lemma \ref{lemmaq-Coeff} $(2)$ and the last equality is derived from Lemma \ref{lemmaq-Coeff} $(4)$.
\end{proof}

\begin{remark}
    By definition of our notation, we expect that $W^R_\Cc(Y,X) = W^0_\Cc(X,Y) + (q^m-1)W^1_\Cc(X,Y)$, where $W^R_\Cc(Y,X)$ is the weight enumerator for the rank metric as defined in \cite[Definition 12]{Shiromoto} for instance.
    Using our Theorem \ref{THMEnumGenRank}, one may compute that $$W^{0}_\Cc(X,Y) + (q^m-1)W^{1}_\Cc(X,Y) = \sum_{t=0}^n\begin{bmatrix}
        n\\t
    \end{bmatrix}_qq^{m(k-t)}\prod_{p=0}^{t-1}(Y-q^pX)X^{n-t}.$$ By looking at the antepenultimate line in proof of Theorem $14$ in \cite{Shiromoto}, we indeed find that $W^R_\Cc(X,Y) = W^0_\Cc(X,Y) + (q^m-1)W^1_\Cc(X,Y)$.
\end{remark}

\begin{example}
    Using the code $\Cc$ given in Example \ref{Exemple2}, we can compute its GRW enumerator for $r=1$. Using computations done in this Example, we can get that $B_0^1(\Cc) = B_1^1(\Cc)=1$, else $B^1_t(\Cc)=0$. This leads to the following GRW enumerator: 
    \[ W^1_\Cc(X,Y) = B_0^1(\Cc)X^3 + B_1^1(\Cc)X^2(Y-X) = X^2Y\]
    The coefficients of this polynomial are consistent with the GRW distribution of $\Cc$ found in Example \ref{Exemple2}.
\end{example}

\section{On the Generalized Rank Weight distribution of MRD codes}\label{Section5}

In this section, we want to characterize the generalized rank weight distribution of an MRD code. We use the notation given in the previous section.

\begin{proposition}\label{dependanceBrtParamMRD}
    Let $\Cc$ be an $[n,k]$ MRD code over $\FF_{q^m}$. We then have \[ B_t^r(\Cc) = \left\{ \begin{array}{cl}
       0  & \text{if } t > k   \\
       \begin{bmatrix}
           n\\ t
       \end{bmatrix}_q\begin{bmatrix}
           k - t \\ r
       \end{bmatrix}_{q^m}  & \text{if } t \leq k
    \end{array}\right.\]
    In particular, $B_t^r(\Cc)$ is completely determined by the parameters of the code, namely $n,m$ and $k$.
\end{proposition}
\begin{proof}
    Since $\Cc$ is MRD, $M_1(\Cc) =n-k +1$. Using that the sequence of generalized rank weights $(M_r(\Cc))_r$ is increasing, we get that $M_r(\Cc) = n-k +r$. Moreover, by \cite[Corollary $41$]{Ravagnani}, the dual code of an MRD code is also MRD. Then, in our case, we have $M_1(\Cc^\perp) = k +1$.

    \noindent We can now substitute these values in Proposition \ref{FormuleBrt} and obtain : $$B_t^r(\Cc) = \left\{ \begin{array}{cl}
       0  & \text{if } t >k -r   \\
       \begin{bmatrix}
           n\\ t
       \end{bmatrix}_q\begin{bmatrix}
           k - t \\ r
       \end{bmatrix}_{q^m}  & \text{if } t \leq k 
    \end{array}\right.$$
    For $t\in \llbracket k-r+1,k\rrbracket$, both quantities in the previous formula are equal. We can finally rewrite $\displaystyle B_t^r(\Cc) = \left\{ \begin{array}{cl}
       0  & \text{if } t >k   \\
       \begin{bmatrix}
           n\\ t
       \end{bmatrix}_q\begin{bmatrix}
           k - t \\ r
       \end{bmatrix}_{q^m}  & \text{if } t \leq k 
    \end{array}\right.$
\end{proof}

Now, we can state our main result.

\begin{theorem}\label{ThmMRDindep}
    Let $\Cc$ be a non-zero $[n,k]$ MRD code of $\FF_{q^m}^n$. The generalized rank weight distribution of $\Cc$ depends only on $n,m$ and $k$. More precisely, we have $$A^r_w(\Cc) = \sum_{t=n-w}^k(-1)^{t+w-n}q^{\frac{(t+w-n)(t+w-n-1)}{2}}\begin{bmatrix}
        t\\n-w
    \end{bmatrix}_q\begin{bmatrix}
        n\\t
    \end{bmatrix}_q\begin{bmatrix}
        k -t\\r
    \end{bmatrix}_{q^m}$$
\end{theorem}
\begin{proof}
    Combining Proposition \ref{AfonctionB} and Proposition \ref{dependanceBrtParamMRD}, we straightforwardly obtain the formula.
\end{proof}

\begin{example}\label{Ex5.3}
    \noindent Consider the $[4,2]$ Gabidulin code $\Cc_1 \subset \FF_{3^4}^4$, see \cite[Definition 2.3.]{HorlemannMRDGabidulin}. Its generator matrix is $G_1 \defegal \begin{pmatrix}
        1 & \alpha & \alpha^2 & \alpha^3 \\ 1 & \alpha^3 & \alpha^6 & \alpha^9
    \end{pmatrix}$, where $\alpha$ is a root of $x^4-x^3-1$ and $\FF_{3^4} = \FF_3(\alpha)$.  Since $\Cc_1$ is Gabidulin, it is MRD and thus the formula of Theorem \ref{ThmMRDindep} gives $A^1_1(\Cc_1) = A^1_2(\Cc_1) = 0$, $A^1_3(\Cc_1) = 40$ and $A^1_4(\Cc_1) = 42$.
    
    \noindent Now, consider the $[4,2]$ code $\Cc_2 \subset \FF_{3^4}^4$ with generator matrix $G_2 \defegal \begin{pmatrix}
        1 & 0 & \alpha & \alpha^2 \\ 0 & 1 & \alpha^2 & 2\alpha
    \end{pmatrix}$ where $\alpha$ is defined previously. This code is MRD but it is neither a Gabidulin code, nor a generalized Gabidulin code, see \cite[Example 5.7.]{HorlemannMRDGabidulin}. In this case, we also find that $A^1_3(\Cc) = 40$, $A^1_4(\Cc)=42$. 

    Thus, with this example, we see that two non-equivalent MRD codes with the same parameters have the same GRW distribution.  
\end{example}

\begin{example}\label{Ex5.4}
    We consider the $[4,3]$ Gabidulin code $\Cc \subset \FF_{3^4}^4$ with generator matrix $G = \begin{pmatrix}
        1 & \alpha & \alpha^2 & \alpha^3 \\ 1 & \alpha^3 & \alpha^6 & \alpha^9 \\ 1 & \alpha^9 & \alpha^{18} & \alpha^{27}
    \end{pmatrix}$, where $\alpha$ is defined by $\alpha^4-\alpha^3-1=0$. Since $\Cc$ is a Gabidulin code, we have that $M_1(\Cc) = 2, M_2(\Cc) =3$ and $M_4(\Cc)= 4$. Moreover, using Theorem \ref{ThmMRDindep}, we obtain $A^1_1(\Cc) = A^2_2(\Cc) = A^3_3(\Cc) = 0$ and $A^1_2(\Cc) = 130, A^1_3(\Cc)=2760, A^1_4(\Cc)=3753,A^2_3(\Cc)=40,A^2_4(\Cc)=6603$ and $A^3_4(\Cc)=1$. These results are consistent since the number of subspaces of $\Cc$ with dimension $1$ and $2$ is respectively $\begin{bmatrix}
        3\\1
    \end{bmatrix}_{3^4}=6643=\begin{bmatrix}
        3\\2
    \end{bmatrix}_{3^4}$.
\end{example}

\section{Acknowledgements}
The author would first like to thank Steven Dougherty for insightful discussions about MacWilliams identities, which were at the very beginnings of this paper. Also, he would like to warmly express his gratitude to Grégory Berhuy and Frédérique Oggier, his PhD supervisors, for their valuable support and their meticulous proofreading. He thanks Pascal Bacchus for his valuable advice on improving the SAGE codes. Finally, the author expresses his sincere thanks to the anonymous referee for his valuable corrections and suggestions.

\bibliographystyle{amsplain}

\end{document}